\documentclass{svmult}
\usepackage{mathptmx,helvet,courier,makeidx,graphicx,epsfig,color,multicol,url}
\usepackage[bottom]{footmisc}
\makeindex 
\usepackage{amsmath,amssymb,booktabs,capt-of,upgreek,enumitem,doi}
\newcommand{\ex}{\mathbb{E}}
\newcommand{\F}{\bar{F}}
\newcommand{\m}{\mathrm{m}}
\newcommand{\h}{\mathrm{h}}
\newcommand{\e}{\mathrm{\ell}}
\newcommand{\g}{\mathrm{g}}
\newcommand{\du}{\mathop{\mathrm{d}u}}
\newcommand{\po}{\uppi}
\renewcommand{\F}{\bar{F}}
\newcommand{\G}{\bar{G}}
\renewcommand{\(}{\left(}
\newcommand{\mrl}{\le_{\text{mrl}}}
\newcommand{\hr}{\le_{\text{hr}}}
\newcommand{\st}{\le_{\text{st}}}
\newcommand{\q}{\mathbf{q}}
\renewcommand{\)}{\right)}
\newcommand{\cx}{\le_{\text{cx}}}
\newcommand{\disp}{\le_{\text{disp}}}
\newcommand{\ew}{\le_{\text{ew}}}
\newcommand{\var}{\operatorname{Var}}
\renewcommand{\)}{\right)}
\makeatletter
\newcommand{\fixed@sra}{$\vrule height 2\fontdimen22\textfont2 width 0pt\implies$}

\usepackage{cleveref}
\Crefname{paragraph}{Paragraph}{Paragraphs}
\Crefname{subsection}{Subsection}{Subsections}
\begin{document}

\title*{Comparative Statics via Stochastic Orderings in a Two-Echelon Market with Upstream Demand Uncertainty}
\titlerunning{Comparative Statics via Stochastic Orderings}
\author{Constandina Koki and Stefanos Leonardos and Costis Melolidakis}
\authorrunning{C. Koki and S. Leonardos and C. Melolidakis} 
\institute{Constandina Koki \at Athens University of Economics and Business, Department of Statistics, Patision 76, 104 34 Athens, Greece \email{kokiconst@aueb.gr}
\and Stefanos Leonardos and Costis Melolidakis \at National \& Kapodistrian University of Athens, Department of Mathematics, Panepistimioupolis, 157 72 Athens, Greece  \email{sleonardos@math.uoa.gr,cmelol@math.uoa.gr}}
\maketitle

\abstract*{We revisit the classic Cournot model and extend it to a two-echelon supply chain with an upstream supplier who operates under demand uncertainty and multiple downstream retailers who compete over quantity. The supplier's belief about retail demand is modeled via a continuous probability distribution function $F$. If $F$ has the \emph{decreasing generalized mean residual life (DGMRL)} property, then the supplier's optimal pricing policy exists and is the unique fixed point of the \emph{mean residual life (MRL)} function. This closed form representation of the supplier's equilibrium strategy facilitates a transparent comparative statics and sensitivity analysis. We utilize the theory of stochastic orderings to study the response of the equilibrium fundamentals -- wholesale price, retail price and quantity -- to varying demand distribution parameters. We examine supply chain performance, in terms of the distribution of profits, supply chain efficiency, in terms of the \emph{Price of Anarchy}, and complement our findings with numerical results.}

\abstract{We revisit the classic Cournot model and extend it to a two-echelon supply chain with an upstream supplier who operates under demand uncertainty and multiple downstream retailers who compete over quantity. The supplier's belief about retail demand is modeled via a continuous probability distribution function $F$. If $F$ has the \emph{decreasing generalized mean residual life (DGMRL)} property, then the supplier's optimal pricing policy exists and is the unique fixed point of the \emph{mean residual life (MRL)} function. This closed form representation of the supplier's equilibrium strategy facilitates a transparent comparative statics and sensitivity analysis. We utilize the theory of stochastic orderings to study the response of the equilibrium fundamentals -- wholesale price, retail price and quantity -- to varying demand distribution parameters. We examine supply chain performance, in terms of the distribution of profits, supply chain efficiency, in terms of the \emph{Price of Anarchy}, and complement our findings with numerical results.}
\vspace{0.2cm}\noindent
\textbf{Keywords:} Continuous Distributions, Demand Uncertainty, Generalized Mean Residual Life, Comparative Statics, Sensitivity Analysis, Stochastic Orders.

\section{Introduction}
The global character of modern markets necessitates the study of competition models that capture two features: first, that retailers' cost is not constant but rather formed as the decision variable of a strategic, profit-maximizing supplier, and second that uncertainty about retail demand affects not only the retailers but also the supplier. Motivated by these considerations, in \cite{Le17}, we use a game-theoretic approach to extend the classic Cournot market in the following two-stage game: in the first-stage (acting as a Stackelberg leader), a revenue-maximizing supplier sets the wholesale price of a product under incomplete information about market demand. Demand or equivalently, the supplier's belief about it, is modeled via a continuous probability distribution. In the second-stage, the competing retailers observe wholesale price and realized market demand and engage in a classic Cournot competition. Retail price is determined by an affine inverse demand function. \par
\cite{La01} studied a similar model in which demand uncertainty affected a single retailer. They identified the property of \emph{increasing generalized failure rate (IGFR)} as a mild unimodality condition for the \emph{deterministic} supplier's revenue-function and then performed an extensive comparative statics and performance (efficiency) analysis of the supply chain at equilibrium. The properties of IGFR random variables were studied in a series of subsequent notes, \cite{La06},\cite{Ba13} and \cite{Pa05}. \par
In \cite{Le17}, we extended the work of \cite{La01} by moving uncertainty to the supplier and by implementing an arbitrary number of second-stage retailers. We introduced the \emph{generalized mean residual life (GMRL)} function of the supplier's belief distribution $F$ and showed that his \emph{stochastic} revenue function is unimodal, if the GMRL function is decreasing -- \emph{(DGMRL)} property -- and $F$ has finite second moment. In this case, we characterized the supplier's optimal price as a fixed point of his \emph{mean residual life (MRL)} function, see \Cref{mainresult} below. Subsequently, we turned our attention to DGMRL random variables, examined their moments, limiting behavior, closure properties and established their relation to IGFR random variables, as in \cite{La06}, \cite{Ba13} and \cite{Pa05}. This study was done in expense of a comparative statistics and performance analysis, as the one in Sections 3-5 of \cite{La01}. The importance of such an analyis is underlined among others in \cite{Ac13},\cite{At02},\cite{Je18} and references therein.
\subsection{Contributions -- Outline:} The present paper aims to fill this gap. Following the methodology of \cite{La01}, we study the response of maket fundamentals by utilizing the closed form characterization of the equilibrium obtained in \cite{Le17}. Specifically, under the conditions of \Cref{mainresult}, the optimal wholesale price is the unique fixed point of the MRL function of the demand distribution $F$. This motivates the study of conditions under which two different markets, denoted by $F_1$ and $F_2$, can be ordered in the mrl-stochastic order, see \cite{Sh07}. \par
The paper is organized as follows. In \Cref{model}, we provide the model description and in \Cref{existing}, the existing results from \cite{Le17} on which the current analysis is based. Our findings, both analytical and numerical are presented in \Cref{comparative}. \Cref{conclusions} concludes our analysis and discusses directions for future work.

\subsubsection{Comparison to Related Works}
Two-echelon markets have been extensively studied in the literature under different perspectives and various levels of demand uncertainty, see e.g. \cite{Be05}, \cite{Pa09}, \cite{Wu12} and \cite{Ya06}. In the present study, we depart from previous works by introducing the toolbox of stochastic orderings in the comparative statics analysis. The advantage of this approach is that we quantify economic notions, such as market size and demand variability, in various ways. Accordingly, we are able to challenge established economic intuitions by showing, for instance, that repsonses of wholesale prices to increasing market size or demand variability are not easy to perdict, since they largely depend on the notion of variability that is employed.

\section{The Model: Game-Theoretic Formulation}\label{model}
An upstream supplier produces a single homogeneous good at constant marginal cost, normalized to $0$, and sells it to a set of $N=\{1,2,\dots, n\}$ downstream retailers. The supplier has ample quantity to cover any possible demand and his only decision variable is the wholesale price $r$ which he determines prior to and independently of the retailers' order-decisions. The retailers observe $r$ -- a price-only contract (there is no return option and the salvage value of the product is zero) -- as well as the market demand parameter $\alpha$ and choose simultaneously and independently their order-quantities $q_i\(r\mid \alpha\), i\in N$. They face no uncertainty about the demand and the quantity that they order from the supplier is equal to the quantity that they sell to the market (at equilibrium). The retail price is determined by an affine inverse demand function $\label{demand}p=\(\alpha-q\(r\)\)^+$, where $\alpha$ is the \emph{demand parameter} and $q\(r\):=\sum_{i=1}^nq_i\(r\)$ is the total quantity that the retailers release to the market\footnote{To simplify notation, we write $q$ or $q\(r\)$ and $q_i$ or $q_i\(r\)$ instead of the proper $q\(r\mid \alpha\)$ and $q_i\(r\mid \alpha\)$.}. Contrary to the retailers, we assume that at the point of his decision, the supplier has incomplete information about the actual market demand. \par
This supply chain can be represented as a two-stage game, in which the supplier acts in the first stage and the retailers in the second. A strategy for the supplier is a price $r\ge 0$ and a strategy for retailer $i$ is a function $q_i:\mathbb R_+\to \mathbb R_+$, which specifies the quantity that retailer $i$ will order for any possible cost $r$. Payoffs are determined via the strategy profile $\(r,\q\(r\)\)$, where $\q\(r\)=\(q_i\(r\)\)_{i=1}^n$. Given cost $r$, the profit function $\po_i\(\q\(r\)\mid r\)$ or simply $\po_i\(\q\mid r\)$, of retailer $i\in N$, is $\po_i\(\q\mid r\)= q_i\(\alpha-q\)^+-rq_i$. For a given value of $\alpha$, the supplier's profit function, $\po_s$ is $\po_s\(r\mid \alpha\)=rq\(r\)$ for $0\le r<\alpha$, where $q\(r\)$ depends on $\alpha$ via $\po_i\(\q\mid r\)$. \par
To model the supplier's uncertainty about retail demand, we assume that after the pricing decision of the supplier, but prior to the order-decisions of the retailers, a value for $\alpha$ is realized from a continuous distribution $F$, with finite mean $\ex\alpha <+\infty$ and nonnegative values, i.e. $F\(0\)=0$. Equivalently, $F$ can be thought of as the supplier's belief about the demand parameter and, hence, about the retailers' willingness-to-pay his price. We will use the notation $\F:=1-F$ for the survival function and $\alpha_L:=\sup{\{r\ge0: F\(r\)=0\}}\ge 0$, $\alpha_H:=\inf{\{r\ge0: F\(r\)=1\}}\le +\infty$ for the support of $F$ respectively.  Under these assumptions, the supplier's payoff function $\po_s$ becomes stochastic: $\po_s\(r\)=\ex \po_s\(r\mid \alpha\)$. All the above are assumed to be common knowledge among the participants in the market (the supplier and the retailers). 

\section{Existing Results}\label{existing}
We consider only subgame perfect equilibria, i.e. strategy profiles $\(r,\q\(r\)\)$ such that $\q\(r\)$ is an equilibrium in the second stage and $q_i\(r\)$ is a best response against any $r$ for all $i=1,2,\dots,n$. The equilibrium behavior of this market has been analyzed in \cite{Le17}. To proceed with the equilibrium representation, we first introduce some notation. 

\subsection{Generalized Mean Residual Life:}Let $\alpha\sim F$ be a nonnegative random variable with finite expectation $\ex \alpha <+\infty$. The \emph{mean residual life (MRL)} function $\m\(r\)$ of $\alpha$ is defined as
\[\m\(r\):=\ex\(\alpha-r \mid \alpha >r\)=\,\dfrac{1}{\F\(r\)}\int_{r}^{\infty}\F\(u\)\du, \quad\mbox{for } r< \alpha_H \] and $\m\(r\):=0$, otherwise, see, e.g., \cite{Be16}, \cite{Lax06} or \cite{Sh07}. In analogy to the \emph{generalized failure rate (GFR)} function $\g\(r\):=r\h\(r\)$, where $\h\(r\):=f\(r\)/\F\(r\)$ denotes the hazard rate of $F$ and the \emph{increasing generalized failure rate (IGFR)} unimodality condition, defined in \cite{La01} and studied in \cite{La06},\cite{Ba13}, we introduce, see \cite{Le17}, the \emph{generalized mean residual life (GMRL)} function $\e\(r\)$, defined as $\e\(r\):=\frac{\m\(r\)}r$, for $0<r<\alpha_H$. If $\e\(r\)$ is \emph{decreasing}, then $F$ has the \emph{(DGMRL) property}. The relationship between the (IGFR) and (DGMRL) classes of random variables is studied in \cite{Le17}. \par
We will use the notation DMRL for a random variable $X$ with a decreasing mean residual life function $\m\(r\)$ and IFR for a random variable $X$ with increasing failure rate $\h\(r\)$. We say that $X_1$ is smaller than $X_2$ in the \emph{mean residual life order}, denoted as $X_1 \mrl X_2$, if $\m_1\(r\)\le \m_2\(r\)$ for all $r$, see \cite{Sh07}. Of course $\m_1\(r\)\le \m_2\(r\)$ if and only if $\e_1\(r\)=\e_2\(r\)$ for all $r>0$. Similarly, $X_1$ is smaller than $X_2$ in the usual stochastic (hazard rate) order, denoted as $X_1\st X_2$ ($X_1\hr X_2$), if $\F_1\(r\)\le \F_2\(r\)$ ($\h_1\(r\)\le \h_2\(r\)$) for all $r$. The $\hr$-order implies the $\mrl$-order. However, neither of the orders $\st$ and $\mrl$ imply the other.

\subsubsection{Market equilibrium:} Using this terminology, we can express the supplier's optimal pricing strategy in terms of the MRL function and formulate sufficient conditions on the demand distribution, under which a subgame perfect equilibrium exists and is unique. 
\begin{theorem}[\cite{Le17}]\label{mainresult}
Assume that the supplier's belief about the unknown, nonnegative demand parameter, $\alpha$, is represented by a continuous distribution $F$, with support inbetween $\alpha_L$ and $\alpha_H$ with $0\le \alpha_L<\alpha_H\le\infty$. 
\begin{enumerate}[label=(\alph*)]
\item If an optimal price $r^*$ for the supplier exists, then $r^*$ satisfies the fixed point equation  
\begin{equation}\label{fixed}r^*=\m\(r^*\)\end{equation} 
\item If $F$ is strictly DGMRL and $\ex \alpha^2$ is finite, then in equilibrium, the optimal price $r^*$ of the supplier exists and is the unique solution of \eqref{fixed}. 
\end{enumerate}
\end{theorem}

Expressing \eqref{fixed} in terms of the GMRL function $\e\(r\)$, the supplier's optimal price $r^*$ can be equivalently written as the solution of equation $\e\(r^*\)=1$. 

\section{Comparative Statics}\label{comparative}
The closed form expression of \eqref{fixed} provides the basis for an extensive comparative statics and sensitivity analysis on the distribution parameters of market demand. To understand the market-equilibrium behavior under different demand (distribution) characteristics, we employ \eqref{fixed} and the rich theory of \emph{stochastic orders}, \cite{Sh07}, \cite{Lax06} and \cite{Be16}. Based on $\po_i$, $\po_s$ and \Cref{mainresult}, the market fundamentals at equilibrium are given in \Cref{fund}.

\begin{center}
$\begin{array}{rclll}
\hline
\multicolumn{3}{l}{\text{Notation \& Expression}} & &\text{Definition} \\\hline
r^* &= & \m\(r^*\) && \text{wholesale price}\\
q^* &=& \frac{n}{n+1}\(\alpha-r^*\)^+ && \text{total quantity sold to the market}\\
p^*&= & \alpha-q^* && \text{retail price}\\
\hline
\Pi_s^*&= & \frac{n}{n+1}\(\alpha-r^*\)^+r^* && \text{realized supplier's profit}\\
\Pi_i^*&= & \(\frac{1}{n+1}\(\alpha-r^*\)^+\)^2 && \text{$i-$th retailer's profit, $i=1,\dots,n$}\\
\hline
\end{array}$
\vspace*{0.2cm}
\captionof{table}{Market fundamentals in equilibrium.}
\label{fund}
\end{center}
Here, $\Pi_s^*$ refers to the supplier's \emph{realized} -- not expected -- profit, i.e. $\Pi_s^*:=\pi_S\(r^*\mid \alpha\)$. From \Cref{fund}, it is immediate that the total quantity $q^*$ that is sold to the market and the retail price $p^*$ are monotone in $r^*$.  Accordingly, we restrict attention on the behavior of $r^*$ as the distribution parameters vary. \par 
To obtain a meaningful comparison between different markets, we assume throughout equilibrium uniqueness. Hence, unless stated otherwise, we consider only strictly DGRML distributions with finite second moment. Since the DGMRL is particulartly inclusive, see \cite{Le17} and \cite{Ba13} and finiteness of the second moment of the demand is naturally to assume, we do not consider them as restrictive. Still, since these conditions are only sufficient and not necessary, the analysis applies to any other setting that guarantees equilibrium existence and uniqueness.

\subsection{Wholesale Price Determinants:}\label{parameters}
Although immediate from \Cref{mainresult}, the next Lemma showcases the importance of the characterization in \eqref{fixed}. Let $X_1$ and $X_2$ denote two markets (or two instances of the same market) with demand distributions $F_1,F_2$. As stated above, $X_1,X_2$ are assumed to be nonnegative, strictly DGMRL random variables with finite second moment. We then have
\begin{lemma}\label[lemma]{mrlorder} Let $X_1,X_2$ denote the demand in two different market instances. If $X_1 \mrl X_2$, then $r^*_1\le r^*_2$. 
\end{lemma}
\begin{proof} Since $X_1\mrl X_2$, we have that $\m_1\(r\)\le \m_2\(r\)$ for all $r>0$, by definition. Hence, by \eqref{fixed}, $r_1^*=\m\(r_1^*\)\le \m_2\(r_1^*\)$, which implies that $\e_2\(r^*_1\)\ge 1$. Since $\e_2\(r\)$ is strictly decreasing by assumption, this implies that $\e_2\(r\)>1$ for all $r<r^*_1$. Since $r^*_2$ is the unique solution of $\e\(r^*_2\)=1$, this in turn implies that $r^*_2\ge r^*_1$.
\end{proof}
Hence, the supplier charges a larger wholesale price in a market that is larger in the $\mrl$-order. Based on \Cref{mrlorder}, the task of studying the behavior of the wholesale price $r^*$ largerly reduces to finding sufficient conditions that imply -- or that are equivalent to -- the $\mrl$-order. Such conditions can be found in \cite{Sh07}, and are studied below.

\subsubsection{Re-estimating Demand:} 
We start with the response of the equilibrium wholesale price $r^*$ to transformations that intuitively correspond to a larger market. Let $X$ denote the random demand in an instance of the market under consideration. Let $c\ge1$ be a positive constant. Moreover, let $Z$ denote an additional source of demand that is independent of $X$. Let $r^*_{X}$ denote the equilibrium wholesale price in the initial market and $r^*_{X+Z}$ the equilibrium wholesale price in the market with random demand $X+Z$. How does $r^*_{X}$ compare to $r^*_{cX}$ and to $r^*_{X+Z}$?\par
While the intuition that the larger markets $cX$ and $X+Z$ will give rise to higher wholesale prices is largely confirmed, see \Cref{reestimate}, the results do not hold in full generality and one needs to pay attention to some technical details. For instance, since DGMRL random variables are not closed under convolution, see \cite{Le17}, the random variable $X+Z$ may not be DGMRL. This may lead to a multiplicity of equilibrium prices in the $X+Z$ instance, irrespectively of whether $X$ is DGMRL or not. To focus on the economic intepretation of the comparative static analysis and to avoid an extensive discussion on the technical conditions, we assume that $Z$ is a random variable such that the market $X+Z$ has again a unique wholesale equilibrium price. However, we consider this assumption as a restriction to the applicability of statement (ii) of \Cref{reestimate}. 

\begin{theorem}\label{reestimate} Let $X\sim F$ be a nonnegative DGMRL random variable with finite second moment which describes the demand distribution in a market instance.
\begin{enumerate}[label=(\roman*)]
\item If $c\ge 1$ is a positive constant, then $r^*_{X}\le r^*_{cX}$.
\item If $Z$ is a nonnegative random variable with finite second moment, independent of $X$ such that $X+Z$ remains strictly DGMRL, then $r^*_{X}\le r^*_{X+Z}$.
\end{enumerate}
\end{theorem}
\begin{proof} The proof of (i) follows directly from the preservation property of the $\mrl$-order that is stated Theorem 2.A.11 of \cite{Sh07}. Specifically, since $c\m\(r/c\)$ is the mrl function of $cX$, we have that for all $>0$ 
\[c\m\(r/c\)=r\cdot\frac{\m\(r/c\)}{r/c}=r\cdot\e\(r/c\)\ge r\cdot\e\(r\)=\m\(r\)\] where the inequality follows from the assumption that $X$ is DGMRL. Hence, $X \mrl cX$ which by \Cref{mrlorder} implies that $r^*_X\le r^*_{cX}$. \par
Statement (ii) is more involved since $r^*_{X+Z}$ may not be unique in general. However, under the assumption that $X+Z$ remains strictly DGMRL, we may adapt Theorem 2.A.11 of \cite{Sh07} and obtain the claim in a similar fashion to part (i). Although Theorem 2.A.11 is stated for DMRL random variables, the proof extends in a straightforward way to DGMRL random variables.
\end{proof}
Another way to treat the possible multiplicity of equilibrium wholesale prices in the $X+Z$ market and the fact that $X+Z$ may not be DGMRL is the following. Since, $X$ is strictly DGMRL, we know that $r<\m_{X}\(r\)$ for all $r<r^*_{X_1}=\m_{X_1}\(r^*_{X_1}\)$. Together with $X_1\mrl X_1+Z$, this implies that for all $r<r^*_{X_1}$, the following holds: $r<\m_{X_1}\(r\)\le \m_{X_1+Z}\(r\)$, and hence that $r^*_{X_1+Z}\ge r^*_{X_1}$ for any $r^*_{X+Z}$ such that \eqref{fixed} holds. Hence, in this case, we can compare the $r^*_{X}$ with every $r^*_{X+Z}$ separately and obtain that $r^*_{X}$ is less than any possible equilibrium wholesale price in the market $r^*_{X+Z}$. However, as mentioned above, we prefer to restrict attention to markets that preserve equilibrium uniqueness.

\subsubsection{Closure Properties:}
Next, we turn our attention to operations that preserve the $\mrl$-order. Let $X_1,X_2$ denote two different instances of the market, i.e., two different demand distributions or beliefs about it, such that $X_1\mrl X_2$. In this case, we know that $r^*_1\le r^*_2$. We are interested in determining transformations of $X_1, X_2$ that preserve the $\mrl$-order and hence, by \Cref{mrlorder}, the ordering $r^*_1\le r^*_2$. Again, to avoid technicalities, we assume that $X_1, X_2$ are such that \Cref{mainresult} applies, i.e., that they are nonnegative, strictly DGMRL and have finite second moment. 

\begin{theorem}\label{closure} Let $X_1\sim F_1,X_2\sim F_2$ denote the demand in two different market instances, such that $X_1\mrl X_2$. Then,
\begin{enumerate}[label=(\roman*)]
\item If $\phi$ is an increasing convex function, then $r^*_{\phi\(X_1\)}\le r^*_{\phi\(X_2\)}$. 
\item If $Z$ is a nonnegative, IFR random variable with finite second moment, independent of $X_1,X_2$ such that $X_1+Z$ and $X_2+Z$ remain strictly DGMRL, then $r^*_{X_1+Z}\le r^*_{X_2+Z}$.
\item If $X_p\sim F_1+\(1-p\)F_2$ is strictly DGMRL for some $p\in \(0,1\)$, then $r^*_{X_1}\le r^*_{X_p}\le r^*_{X_2}$.
\end{enumerate}
\end{theorem}
\begin{proof} Statements (i) through (iii) follow directly from Theorems 2.A.19, Lemma 2.A.8 and Theorem 2.A.18 respectively. The assumption that the transformed random variables remain strictly DGMRL ensures equilibrium uniqueness.
\end{proof}
If instead of $X_1\mrl X_2$, $X_1$ and $X_2$ are ordered in the weaker $\hr$-order, i.e., if $X_1\hr X_2$ and $Z$ is DMRL (instead of merely IFR), then Lemma 2.A.10 of \cite{Sh07} implies that statement (ii) of \Cref{closure} remains true. Formally,
\begin{corollary}Let $X_1\sim F_1,X_2\sim F_2$ denote the demand in two different market instances, such that $X_1\hr X_2$. If $Z$ is a nonnegative, IFR random variable with finite second moment, independent of $X_1,X_2$ such that $X_1+Z$ and $X_2+Z$ remain strictly DGMRL, then $r^*_{X_1+Z}\le r^*_{X_2+Z}$.
\end{corollary}
Following the exposition of \cite{Sh07}, the above collection of statements can be extended to incorporate more case-specific results. \par
Although \Cref{reestimate,closure} are immediate \emph{once} \Cref{mainresult} and \Cref{mrlorder} have been established, their implications in terms of the economic intuitions are non-trivial. In particular, both Theorems imply that if the supplier reestimates upwards her expectations about the demand then she will charge a higher price. However, this intuitive conclusion depends on the conditions that imply the $\mrl$-order and does not hold in general, as discussed in \Cref{larger} below. 

\subsubsection{Market Demand Variability:}
The response of the equilibrium wholesale price to increasing (decreasing) demand variability is less straightforward. There exist several notions of stochastic orders that compare random variables in terms of their variability and depending on which we employ, we may derive different results. First, we introduce some notation. \\[0.3cm]
Variability or Dispersive Orders: Let $X_1\sim F_1$ and $X_2\sim F_2$ be two nonnegative random variables with equal means, $\ex X_1=\ex X_2$, and finite second moments. If $\int_{r}^{+\infty}\F_1\(u\)\du\le \int_{r}^{+\infty}\F_2\(u\)\du$ for all $r\ge0$, then $X_1$ is said to be smaller than $X_2$ in the \emph{convex order}, denoted by $X_1\cx X_2$. If $F_1^{-1}$ and $F_2^{-1}$ denote the right continuous inverses of $F_1,F_2$ and $F_1^{-1}\(r\)-F_1^{-1}\(s\)\le F_2^{-1}\(r\)-F_2^{-1}\(s\)$ for all $0<r\le s<1$, then $X_1$ is said to be smaller than $X_2$ in the \emph{dispersive order}, denoted by $X_1\disp X_2$. Finally, if $\int_{F_1^{-1}\(p\)}^{\infty} \F_1\(u\)\du \le \int_{F_2^{-1}\(p\)}^{\infty} \F_2\(u\)\du$ for all $p\in \(0,1\)$, then $X_1$ is said to be smaller than $X_2$ in the \emph{excess wealth order}, denoted by $X_1\ew Y$. \cite{Sh07} show that $X\disp Y \implies X\ew Y \implies X\cx Y$ which in turn implies that $\var\(X\)\le \var\(Y\)$. Further insights and motivation about these orders are provided in Chapter 3 of \cite{Sh07}.\\[0.3cm]
Less Variability implies Lower Wholesale Price: Under our assumptions the $\mrl$-order is not implied by the $\cx$-order. Hence, the $\cx$-order is not enough to conclude that wholesale prices are ordered according to the respective market variability, i.e., that less (more) variability gives rise to lower (higher) wholesale prices. However, if we restrict attention to the $\ew$ and $\disp$ orders, then more can be said. Recall that $\alpha_L$ denotes the left end of the support of a variable $X$. Accordingly, we will write $\alpha_{Li}$ to denote the left end of the support of variable $X_i$ for $i=1,2$. 

\begin{theorem}\label{ewealth}
Let $X_1\sim F_1, X_2\sim F_2$ be two nonnegative, DGMRL random variables with $\alpha_{L1}\le \alpha_{L2}$ which denote the demand in two different market instances. If either $X_1$, $X_2$ or both are DMRL and $X_1\ew X_2$, then $r^*_1\le r^*_2$.
\end{theorem}
\Cref{ewealth} follows directly from Theorem 3.C.5 of \cite{Sh07}. Based on its proof, the assumption that at least one of the two random variables is DMRL (and not merely DGMRL) cannot be relaxed. \cite{Be16} argue about the restricted applicability of the $\ew$-order due to the difficulty in the evaluation of incomplete integrals of quantile functions and provide useful characterizations of the $\ew$-order to remedy this problem. \par
A result of similar flavor can be obtained if we use the $\disp$ order instead. Again, the condition that both $X_1$ and $X_2$ are DGMRL does not suffice and we need to assume that at least one is IFR. 
\begin{theorem}\label{dispersive}
Let $X_1\sim F_1, X_2\sim F_2$ be two nonnegative, DGMRL random variables which denote the demand in two different market instances. If either $X_1$, $X_2$ or both are IFR and $X_1\disp X_2$, then $r^*_1\le r^*_2$.
\end{theorem}
\Cref{dispersive} follows directly from Theorem 3.B.20 (b) of \cite{Sh07} and the fact that the $\hr$-order implies the $\mrl$-order. Again, more case specific results can be drawn from the analysis of \cite{Sh07}. \par
The main insight that we get from \Cref{ewealth,dispersive} is that less (more) variability implies lower (higher) wholesale prices. This is in sharp contrast with the results of \cite{La01} and sheds light on the effects of demand uncertainty. If uncertainty affects the retailer, then the supplier charges a higher price and captures all supply chain profits as variability reduces. Contrarily, if uncertainty falls to the supplier, then the supplier charges a lower price as variability increases. In this case, the supplier captures a lower share of system profit, see also \eqref{ratio} below. \par
The above cases correspond to two extremes: in \cite{La01} uncertainty falls solely to the retailer, whereas in the present analysis uncertainty falls solely to the supplier. Based on the aggregate findings for the two cases, the following question naturally arises: is there a way to distribute demand uncertainty among supplier and retailers to mitigate its adverse effects and to evenly distribute supply chain profits among market participants? Answering this question exceeds the scope of the comparative statics analysis. However, it highlights a interesting direction for future work. \\[0.3cm]
Parametric families of distributions: To further elaborate on the effects of relative variability on the wholesale price, we may compare our analysis with the approach of \cite{La01}. Given a random variable $X$ with distribution $F$, \cite{La01} consider the random variables $X_i:=\delta_i+\lambda_iX$ with $\delta_i\ge0$ and $\lambda_i>0$ for $i=1,2$. They conclude that in this case, the wholesale price is dictated by the coefficient of variation, $CV_i=\frac{\sqrt{\var\(X_i\)}}{\ex X_i}$. Specifically, if $CV_2<CV_1$, then $r^*_1<r^*_2$, i.e., in their model, a lower CV, or equivalently a lower relative variability, implies a higher price.\par
To establish a similar comparison, we utilize the comprehensive comparison of parametric distributions in terms of stochastic orders that is provided in Section 2.9 of \cite{Be16}. For instance, consider two normal random variables $X_1\sim N\(\mu_1,\sigma_1^2\)$ and $X_2\sim N\(\mu_2,\sigma_2^2\)$. By Table 2.2 of \cite{Be16}, if $\sigma_1<\sigma_2$ and $\mu_1\le \mu_2$, then $X_1\mrl X_2$ and hence, by \Cref{mrlorder}, $r^*_1<r^*_2$. However, by choosing $\sigma_i$ and $\mu_i$ appropriately, we can achieve an arbitrary ordering of relative variability, i.e. of $CV_1$ and $CV_2$. The reason is that the conclusions from this approach are obscured by the fact that changing $\mu_i$ for $i=1,2$, does not only affect the $CV_i$'s but also the central location of the demand distributions. In this sense, the approach using dispersive orders seems more appropriate because, under the assumption that $\ex X_1=\ex X_2$, it isolates the effect of the variability of the distribution on wholesale prices via stochastic orderings.

\subsubsection{Stochastically Larger Market:}\label{larger}
It is well known that the usual $\st$-order does not imply nor is implied by the $\mrl$-order, see \cite{Sh07}. This implies that in a stochastically larger market, the supplier may still charge a lower price, which is in line with the intuition of \cite{La01} that ``size is not everything'' and that prices are driven by different forces. Such an example is provided below. Let 
\[f\(r;\omega,\kappa,\phi\):=\frac{\kappa \(\kappa^2+\omega^2\)}{\kappa^2\cos\(\omega\phi\)+\kappa^2+\kappa\omega\sin\(\omega\phi\)+\omega^2} \cdot e^{-\kappa r}\(\cos\(\omega\(r-\phi\)\)+1\)\]
for $r\ge0$, denote the densities of a parametric family of exponentially decaying sinusoids. For $\(\omega,\kappa,\phi\)=\(0,\kappa,0\)$, $f$ corresponds to the exponential distribution with parameter $\kappa$. \Cref{inverse} depicts the survival functions $\F, \G$, the log-survival ratio $\log{\(\F/\G\)}$ and the optimal wholesaleprices $r^*_F$ and $r^*_G$ for $F$ corresponding to $\(\omega,\kappa,\phi\)=\(\pi,0.8,1.2\)$ and $G$ to $\(\omega,\kappa,\phi\)=\(0,0.9,0\)$. 
\begin{figure}[ht!]
\centering
\includegraphics[width=\linewidth]{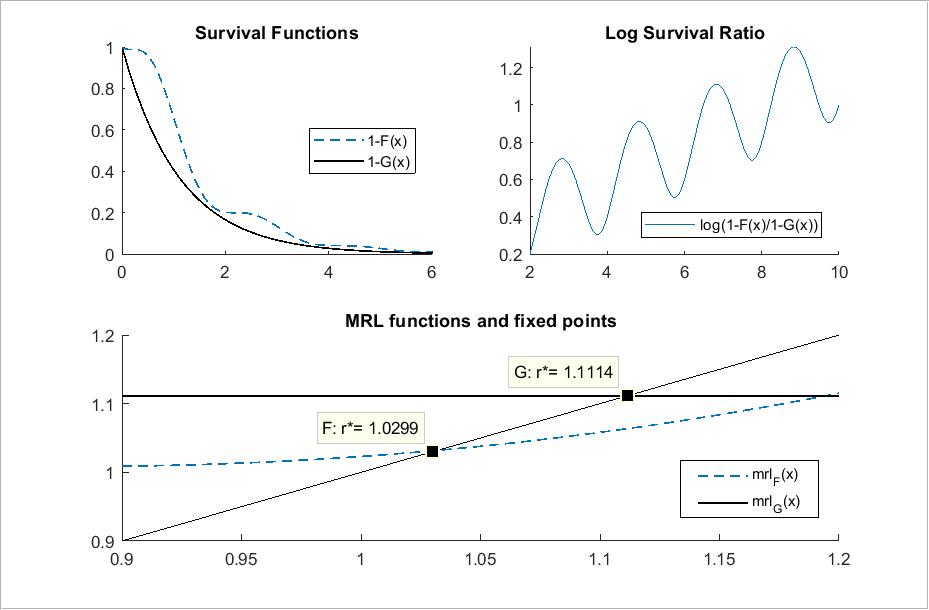}
\caption{$F$ stochastically dominates $G$, however $r^*_G>r^*_F$.}
\label{inverse}
\end{figure}
Since the log-survival ratio remains throughout positive, we infer that $G\st F$. However, as shown in the graph below $r^*_F=1.0299<r^*_G=1.1114$. Although, both functions have a unique fixed point, $F$ is not DMRL (nor DGMRL). Several simulations have not provided a conclusive answer to whether stochastic dominance implies also a larger price if we restrict to the DMRL (or DGMRL) class of random variables.

\subsection{Supply Chain Performance:}\label{performance} We measure the supply chain performance in terms of the ratio $\sum_{i=1}^n\Pi_i^*/\Pi_s^*$ which describes the division of the \emph{realized} system profit between retailers and suppliers. If $\alpha\le r^*$, then there is no transaction and the profits of all participants are equal to zero. For $\alpha>r^*$, we have that 
\begin{equation}\label{ratio}\frac{\sum_{i=1}^n \Pi^*_i}{\Pi^*_s}=\frac{n\(\frac{1}{n+1}\(\alpha-r^*\)\)^2}{\frac{n}{n+1}r^*\(\alpha-r^*\)}=\frac1{n+1}\(\frac{\alpha}{r^*}-1\)\end{equation}
Hence, the division of realized profit between supplier and retailers depends on the number $n$ of retailers and the wholesale price $r^*$. Specifically, for a given realized demand $\alpha$, as $n$ or $r^*$ increase, the supplier captures a larger share of the system profits. 

\subsubsection{Supply Chain Efficiency:}\label{efficiency} As a benchmark, we will first determine the equilibrium behavior and performance of an integrated supply chain. Let $\po_I$ denote the profit of an integrated firm. The integrated firms' decision variable is now the retail price $r$, and hence its expected profit is given by $\po_I\(r\)=r\ex\(\alpha-r\)^+=r\m\(r\)\F\(r\)$. By the same argument as in the proof of \Cref{mainresult}, $\po_I$ is maximized at $r^*=\m\(r^*\)$. In particular, the equilibrium price of both the integrated and non-integrated supplier is the same. Hence, the integrated firm's realized profit in equilibrium is equal to $\Pi_I^*\(r^*\mid \alpha\)=r^*\(\alpha-r^*\)^+$. \par
In a similar fashion to \cite{Pe07}, we define the realized \emph{Price of Anarchy (PoA)} of the system as the worst-case ratio of the realized profit of the centralized supply chain, $\Pi_I^*$, to the realized aggregate profit of the decentralized supply chain, $\Pi^*_D:=\Pi^*_s+\sum_{i=1}^n \Pi^*_i$. To retain equilibrium uniqueness, we restrict attention to the class $\mathcal G$ of nonnegative DGMRL random variables. If the realized demand $\alpha$ is less than $r^*$, then both the centralized and decentralized chain make $0$ profits. Hence, we define the PoA as: $\text{PoA}:=\sup_{F\in \mathcal G}\sup_{\alpha>r^*} \frac{\Pi^*_I}{\Pi^*_D}$. We then have
\begin{theorem}\label{poathm} The PoA of the system is given by  
\begin{equation}\label{poa}\text{PoA}=1+\mathcal O\(1/n\)\end{equation}
\end{theorem} 
\begin{proof} A direct substitution in the definition of PoA yields: 
\begin{equation}\label{poacalc}\text{PoA}:=\sup_{F\in \mathcal G}\sup_{\alpha>r^*}\left\{\frac{\(n+1\)^2}{n}\cdot\(n+\frac{\alpha}{r^*}\)^{-1}\right\}\end{equation}
Since $\(n+\frac{\alpha}{r^*}\)^{-1}$ decreases in the ratio $\alpha/r^*$, the inner $\sup$ is attained asymptotically for $\alpha \to r^*$. Hence, $\text{PoA}=\sup_{F\in \mathcal G}\left\{\frac{\(n+1\)^2}{n}\cdot \(n+1\)^{-1}\right\}=1+\frac1n$. \end{proof}
\Cref{poathm} implies that the supply chain becomes less efficient as the number of downstream retailers increases. Although the PoA provides a useful worst-case scenario, for a fixed $F$ and a realized demand $\alpha$, it is also of interest to study the response of the ratio $\Pi^*_I/\Pi^*_D$ to different wholesale prices. Specifically, for any given value of $\alpha$, and fixed $F$, $\Pi^*_I/\Pi^*_D$ increases as the wholesale price increases. Hence, a higher wholesale price corresponds to worst efficiency for the decentralized chain. \par
Together with the observation that with a higher wholesale price, the supplier captures a larger share of the system profits, this motivates -- from a social perspective -- the study of mechanisms that will lead to reduced wholesale prices for fixed demand levels and fixed market characteristics (number of retailers and demand distribution). Such a study falls not within the topic of the present analysis but constitutes a promising direction for future research.

\section{Conclusions}\label{conclusions}
Along with \cite{Le17}, the present study provides a probabilistic and economic analysis that aims to extend the work of \cite{La01}, \cite{La06}, \cite{Pa05} and \cite{Ba13}.\footnote{The current paper and \cite{Bel18} contain preliminary results that appear in full length in \cite{Le18,Leo20,Leo21,Leon21}.} The characterization, under mild conditions, of the supplier's optimal pricing policy as the unique fixed point of the MRL function of the demand distribution, provides a powerful tool for a multifaceted comparative statics analysis. \Cref{reestimate,closure} demonstrate how stochastic orderings, coupled with this characterization, provide predictions of the response of the wholesale price in a versatile environment of various demand transformations. Based on a numerical example, \Cref{larger} confirms \cite{La01}'s intuition that prices are driven by different forces than market size. In \Cref{performance} and \Cref{efficiency}, we show that number of second stage retailers and wholesale prices have a direct impact on supply chain performance and efficiency. A more extended version of the present study is subject of ongoing work.

\end{document}